\documentclass[aps,pra,reprint,amsmath,amssymb,longbibliography]{revtex4-2}

\usepackage{amsthm}
\usepackage{mathtools}
\usepackage{microtype}
\usepackage{xcolor}
\usepackage{hyperref}
\usepackage[capitalise,noabbrev]{cleveref}

\hypersetup{
  colorlinks=true,
  linkcolor=blue,
  citecolor=green!50!black,
  urlcolor=blue
}

\newtheorem{theorem}{Theorem}
\newtheorem{lemma}{Lemma}

\newtheorem{corollary}{Corollary}

\newcommand{\tr}{\operatorname{tr}}

\newcommand{\cH}{\mathcal H}
\newcommand{\cS}{\mathcal S}
\newcommand{\ket}[1]{\lvert #1\rangle}
\newcommand{\bra}[1]{\langle #1\rvert}
\newcommand{\braket}[2]{\langle #1\mid #2\rangle}
\newcommand{\norm}[1]{\left\lVert #1\right\rVert}
\newcommand{\opnorm}[1]{\left\lVert #1\right\rVert}
\newcommand{\Tr}{\operatorname{Tr}}
\newcommand{\Span}{\operatorname{span}}
\newcommand{\diag}{\operatorname{diag}}

\begin{document}

\title{Small-Bias Quantum Approximate Counting via the Multiplicative Adversary Method}
\author{Albert Lin}
\email{albertlin2468@gmail.com}

\author{Han-Hsuan Lin}
\email{linhh@cs.nthu.edu.tw}

\affiliation{
National Center for Excellence in Quantum Information Science and Engineering,
National Tsing Hua University,
Hsinchu, Taiwan
}

\begin{abstract}
We study the two-weight decision version of quantum approximate counting:
given oracle access to \(x\in\{0,1\}^N\), distinguish
\[
|x|=M
\qquad\text{from}\qquad
|x|=M+\Delta
\]
with success probability \(1/2+\zeta\). Using the multiplicative adversary
method, we prove
\[
\Omega\left(
\max\left\{
\zeta\frac{\sqrt{(N-M)(M+\Delta)}}{\Delta},
\sqrt{\frac{\zeta N}{\Delta}}
\right\}
\right).
\]
The same parameter dependence follows from the polynomial-method
characterization of the two-layer symmetric function by Podder, Yao, and
Ye. Our contribution is a multiplicative-adversary derivation that keeps
track of the progress produced by individual oracle queries.

For the first term, after complementing the input if necessary, we assume
\(M+\Delta\le N-M\). We use the Hamming-layer subspaces from the
eigenspace method of Ambainis, \v{S}palek, and de Wolf and compose their
adjacent-layer unitary maps to relate the two nonadjacent promise layers.
After fixing the queried coordinate, the analysis block-diagonalizes into
four-dimensional subspaces. An exact calculation of the corresponding
one-query progress ratio gives the first lower bound. The
same estimate also implies
\[
\left\|
(I-\widehat\Pi_{\mathrm{bad}})\ket{\Psi^T}
\right\|^2
=
O\left(
\frac{T^2\Delta^2}{(N-M)(M+\Delta)}
\right)
\]
for the coherent input superposition used in the adversary argument.

For the second term, we prove directly by a three-eigenvalue
multiplicative adversary that unique OR on \(n\) bits with success
probability \(1/2+\zeta\) requires \(\Omega(\sqrt{\zeta n})\) queries,
and then reduce unique OR to the two-weight counting problem.
\end{abstract}

\maketitle

\section{Introduction}

Quantum query lower bounds are fundamental to understanding the computational
power and limitations of quantum computers.  Most standard lower bounds are
formulated in the bounded-error regime \cite{AKKT20,NayakWu99,Amb02}, where
an algorithm is required to succeed with some fixed constant probability,
such as \(2/3\).  In a number of settings, however, it is useful to understand
computation before this constant-success threshold is reached.  This leads
naturally to the \emph{small-bias regime}, in which, for decision problems,
the success probability is
\[
\frac12+\zeta,
\]
where \(\zeta\) may be arbitrarily small
\cite{ASW06,Spa07,PodderYY25}, and one studies the quantitative relation
between the number of quantum queries and the achievable advantage over
random guessing.

Several techniques have been developed for proving quantum query lower
bounds, with different behavior in the small-bias regime.  The polynomial
method originates from the use of real polynomials for Boolean functions
\cite{NisanSzegedy94} and was adapted to quantum query algorithms by Beals,
Buhrman, Cleve, Mosca, and de Wolf \cite{Beals01}.  For symmetric Boolean
functions, approximation results such as Paturi's theorem give particularly
sharp degree bounds \cite{Paturi92}.  The polynomial method applies naturally
to small success probabilities and has been used to obtain fine-grained
query--bias tradeoffs for symmetric functions \cite{PodderYY25}.  The
adversary method was introduced by Ambainis \cite{Amb02} and developed in
weighted and spectral forms \cite{HNS02,BSS03}, whose equivalence was
clarified by \v{S}palek and Szegedy \cite{SpalekSzegedy06}.  While the
standard adversary method is most commonly used in the bounded-error setting,
the eigenspace method \cite{ASW06} and the multiplicative adversary method
\cite{Spa07} provide tools for controlling smaller success probabilities.
Another important technique is the compressed-oracle method, introduced by
Zhandry \cite{Zhandry19}, which is particularly effective for random-oracle
and cryptographic query problems and can handle very small success
probabilities.  Recent work of Jeffery and Zur relates the compressed-oracle
method to a restricted multiplicative-adversary framework
\cite{JefferyZur26}.

One of our motivations for studying the small-bias regime comes from NISQ
and bounded-depth quantum computations \cite{ChenCHL22,HamoudiLS24}.
NISQ algorithms can be modeled as hybrid computations in which a classical
computer repeatedly invokes noisy quantum circuits and processes their
classical outcomes \cite{ChenCHL22}.  Under a restriction on coherent depth,
an individual quantum computation may be unable to solve the underlying
problem with bounded error while still producing a nonzero distinguishing
advantage.  A bounded-error lower bound remains valid in such a model, but by
itself does not quantify this partial progress.  Small-bias bounds provide a
finer description of what can be achieved by such a restricted quantum
computation.  Related questions arise, for example, in the NISQ complexity
of collision finding, where tight bounds have been obtained for several
hybrid and bounded-depth quantum models \cite{HamoudiLS24}.

Small-success-probability bounds also arise naturally in the study of
direct-product theorems, parallel computation, and time--space tradeoffs
\cite{Amb05,KSW07,ASW06,Spa07}.  Strong direct-product theorems study
algorithms that are given substantially fewer resources than would be
required to solve many independent instances reliably; in this regime, the
relevant conclusion is that their success probability becomes exponentially
small rather than merely failing to exceed a fixed constant.  For quantum
search, strong direct-product behavior was first established through
polynomial and adversary-based arguments \cite{Amb05,KSW07}.  Ambainis,
\v{S}palek, and de Wolf then developed an eigenspace-based adversary method
for symmetric functions and used it to obtain strong direct-product theorems
and quantum time--space tradeoffs \cite{ASW06}.  The multiplicative adversary
method was subsequently formulated to control multiplicative progress
directly and to remain effective when the success probability is
exponentially small, a regime in which the ordinary additive adversary bound
can become negligible \cite{Spa07}.  Small-success-probability lower bounds
also arise in collision problems; for example, Hamoudi and Magniez employ
bounds in an exponentially small success-probability regime to derive
quantum time--space tradeoffs for finding multiple collision pairs
\cite{HamoudiMagniez23}.

A further motivation comes from post-quantum cryptography
\cite{Zhandry19,JefferyZur26}.  Cryptographic security requires quantitative
control of an adversary's success probability or distinguishing advantage
even when that advantage is subconstant or negligible.  Collision finding
provides a canonical example because collision-resistant hashing is a basic
cryptographic primitive and quantum algorithms substantially change the
query complexity of collision attacks.  More generally, the
compressed-oracle technique was developed to analyze quantum adversaries in
the quantum random-oracle model and has become an important tool in
post-quantum security proofs \cite{Zhandry19}.

\medskip

Quantum approximate counting is a fundamental problem in quantum query
complexity, as it generalizes unstructured search from finding a marked item
to estimating the number of marked items.  Furthermore, the problem is
closely related to amplitude estimation, a central primitive in quantum
algorithms.  Recent work has also developed simplified approaches to quantum
approximate counting and amplitude estimation that avoid the use of quantum
phase estimation \cite{AaronsonRall20}.

We study the quantum query complexity of approximate counting in the
small-bias regime.  Given oracle access to an input string
\(x\in\{0,1\}^N\), we consider the two-weight decision problem
\[
|x|=M
\qquad\text{or}\qquad
|x|=M+\Delta.
\]
Equivalently, this is the two-layer symmetric promise problem
\(f_N^{M,M+\Delta}\).  In the usual relative-gap parametrization,
\[
\Delta=\epsilon M,
\]
so that the second promised weight is \((1+\epsilon)M\).

For bounded-error algorithms, quantum counting and amplitude-estimation
techniques give the standard upper bound \cite{Bra02}, with a later QFT-free
construction achieving the same query complexity using only Grover
iterations \cite{AaronsonRall20}.  On the lower-bound side, Nayak and Wu
used the polynomial method to derive lower bounds for approximate counting
and related statistics \cite{NayakWu99}.  In the usual regime
\[
\Delta=\epsilon M\le M,
\qquad
M+\Delta\le N-M,
\]
the bounded-error query complexity is
\[
\Theta\left(
\frac{1}{\epsilon}\sqrt{\frac{N}{M}}
\right).
\]
More recent work has extended the polynomial approach to approximate
counting to stronger resource models, including Laurent-polynomial lower
bounds that establish tradeoffs between query complexity and additional
quantum resources \cite{AKKT20}.

Recently, Podder, Yao, and Ye \cite{PodderYY25} characterized the
fine-grained query complexity of two-layer symmetric functions.
Specializing their result to the two weights considered here gives
\begin{equation}
\Theta\left(
\max\left\{
\zeta
\frac{\sqrt{(N-M)(M+\Delta)}}{\Delta},
\sqrt{\frac{\zeta N}{\Delta}}
\right\}
\right),
\label{eq:intro-known-bound}
\end{equation}
up to the trivial constant-query floor.  In the relative-gap regime
\[
\Delta=\epsilon M\le M,
\qquad
M+\Delta\le N-M,
\]
this becomes
\[
\Theta\left(
\max\left\{
\frac{\zeta}{\epsilon}\sqrt{\frac{N}{M}},
\sqrt{\frac{\zeta N}{\epsilon M}}
\right\}
\right).
\]

The purpose of the present work is not to improve the known asymptotic lower
bound.  Rather, we ask whether the same fine-grained behavior can be
understood directly through the multiplicative adversary method.  Our first
contribution is a direct MADV analysis of the two promised Hamming-weight
layers.  After fixing the queried coordinate, the action of one oracle query
is reduced to a \(4\times4\) block.  Directly bounding the one-query
progress ratio
\[
\left\|
G_j^{1/2}WG_j^{-1/2}
\right\|^2
\]
yields
\[
\Omega\left(
\zeta
\frac{\sqrt{(N-M)(M+\Delta)}}{\Delta}
\right).
\]
In contrast to the polynomial approach, this formulation describes the
lower bound through the evolution of a progress quantity under individual
queries.

The same one-query estimate also gives a quantitative statement about the
coherent adversary state.  If \(\widehat\Pi_{\mathrm{bad}}\) denotes the
extension of the projector onto the low-eigenvalue subspace of the adversary
matrix, then after \(T\) queries,
\[
\left\|
(I-\widehat\Pi_{\mathrm{bad}})
\ket{\Psi^T}
\right\|^2
=
O\left(
\frac{T^2\Delta^2}
{(N-M)(M+\Delta)}
\right).
\]
Thus the same calculation bounds how rapidly the coherent state used in the
adversary argument can develop components outside the low-eigenvalue
subspace.

Our second contribution gives an independent MADV derivation of the other
term in \eqref{eq:intro-known-bound}.  We prove directly that unique OR on
\(n\) bits with success probability \(1/2+\zeta\) requires
\[
\Omega(\sqrt{\zeta n})
\]
queries.  The proof uses a three-level multiplicative adversary and a direct
lower bound on the final progress \(W_T\), rather than the bad/good-space
criterion used for the direct counting argument.  A reduction from unique
OR to approximate counting then yields
\[
\Omega\left(
\sqrt{\frac{\zeta N}{\Delta}}
\right).
\]
Combining the two arguments recovers the known optimal fine-grained
small-bias lower bound using multiplicative-adversary methods.

We hope that this viewpoint, together with the explicit query-level analysis
developed here, will be useful in settings where understanding the evolution
of quantum distinguishability is important, including the study of NISQ
complexity, quantum cryptographic applications, and related
resource-constrained models.

\subsection{Proof overview}

Our proof establishes the two lower bounds
\[
\Omega\left(
\zeta
\frac{\sqrt{(N-M)(M+\Delta)}}{\Delta}
\right)
\qquad\text{and}\qquad
\Omega\left(
\sqrt{\frac{\zeta N}{\Delta}}
\right)
\]
separately.  We first obtain the first term through a direct
multiplicative-adversary analysis of the two promised Hamming-weight layers,
and then derive the second term from an independent small-bias lower bound
for unique OR.

For the direct counting argument derived in~\ref{sec:counting-direct}, we first orient the two promise layers,
using bitwise complementation if necessary, so that
\[
M+\Delta\le N-M.
\]
The quantity appearing in the first lower bound is invariant under this
transformation, so this orientation does not change the resulting bound.

The construction builds on the eigenspace decomposition for symmetric
functions used by Ambainis, \v{S}palek, and de Wolf
\cite{ASW06}.  Their analysis treats adjacent Hamming-weight layers.
Here the two promised weights may be separated by an arbitrary gap
\(\Delta\), so we compose the adjacent-layer unitary maps to identify the
corresponding subspaces on the two Hamming-weight layers.

After fixing the queried coordinate, these identifications give an
orthogonal decomposition into four-dimensional subspaces.  On each
nontrivial block, the action of one oracle query can be written explicitly,
and the multiplicative change of the adversary progress is controlled by
analyzing the corresponding \(4\times4\) block.  This direct block analysis
yields the one-query bound needed to obtain
\[
\Omega\left(
\zeta
\frac{\sqrt{(N-M)(M+\Delta)}}{\Delta}
\right).
\]
The same one-query estimate also gives a quantitative bound on the final
weight outside the low-eigenvalue subspace of the adversary matrix:
\[
\left\|
(I-\widehat\Pi_{\mathrm{bad}})
\ket{\Psi^T}
\right\|^2
=
O\left(
\frac{T^2\Delta^2}
{(N-M)(M+\Delta)}
\right).
\]
Thus the direct MADV argument not only recovers the first query lower bound
but also describes how quickly distinguishing components can develop during
the computation.

We next derive the second term through unique OR.  In
Section~\ref{sec:uor}, we give a direct multiplicative-adversary proof that
unique OR on \(n\) bits with success probability \(1/2+\zeta\) requires
\[
\Omega(\sqrt{\zeta n})
\]
queries.  This proof uses a three-level adversary matrix and derives the
required final progress directly from the success condition, rather than
using the bad/good-space criterion of the counting argument.

Finally, in Subsection~\ref{sec:uor-to-counting}, we reduce unique OR to the
two-weight counting problem by repeating the unknown OR input
\(\Delta\) times and appending known bits so that the two cases have Hamming
weights \(M\) and \(M+\Delta\).  Choosing the unique-OR instance size to be
of order \(N/\Delta\) gives
\[
\Omega\left(
\sqrt{\frac{\zeta N}{\Delta}}
\right).
\]
Combining the direct Hamming-layer argument with the unique-OR reduction
therefore yields
\[
\Omega\left(
\max\left\{
\zeta
\frac{\sqrt{(N-M)(M+\Delta)}}{\Delta},
\sqrt{\frac{\zeta N}{\Delta}}
\right\}
\right),
\]
recovering the known fine-grained small-bias lower bound of
\cite{PodderYY25} through multiplicative-adversary arguments.

The remainder of the paper is organized as follows.
Section~\ref{sec:madv-framework} states the common multiplicative-adversary
framework.
Section~\ref{sec:counting-direct} contains the direct MADV analysis of
approximate counting.
Section~\ref{sec:uor} proves the small-bias unique-OR lower bound.
Section~\ref{sec:reduction-final} gives the reduction from unique OR to
approximate counting and combines the two lower bounds.

\section{Multiplicative adversary framework}
\label{sec:madv-framework}

Both promise problems considered in this work have Boolean output, so we state only the Boolean-output version of the multiplicative adversary method. We use the phase-oracle formulation, which is equivalent to the standard bit-query formulation up to input-independent unitaries.

Consider a Boolean promise problem
\[
g:X=X_0\cup X_1\longrightarrow\{0,1\},
\qquad
g(x)=z\quad(x\in X_z),
\]
on \(n\) input bits.  Let
\[
\mathcal H_I=\Span\{\ket{x}:x\in X\}
\]
be an inaccessible input-label register and let \(\mathcal H_A\) contain the
query register, workspace, and output register of the algorithm.

For \(i\in[n]\), define the input-side phase operator
\[
O_i
=
\sum_{x\in X}
(-1)^{x_i}\ket{x}\!\bra{x}.
\]
For a normalized input superposition
\[
\ket{\delta}
=
\sum_{x\in X}\delta_x\ket{x},
\]
run the algorithm coherently.  After \(t\) queries the joint state has the
form
\[
\ket{\Psi^t}
=
\sum_{x\in X}
\delta_x\ket{x}\ket{\psi_x^t}.
\]
Its reduced state on the input register is
\[
\rho_I^t
=
\tr_{\mathcal H_A}
\ket{\Psi^t}\!\bra{\Psi^t}.
\]

\subsection{Query progress}
Let \(\Gamma\succ0\) be an adversary matrix whose smallest eigenvalue is
\(1\), and let \(\ket{\delta}\) be a normalized eigenvector satisfying
\[
\Gamma\ket{\delta}=\ket{\delta}.
\]
For operators on the input register, write
\[
\langle A,B\rangle:=\Tr(A^\dagger B).
\]
Define
\begin{equation}
W_t
=
\langle\Gamma,\rho_I^t\rangle
=
\bra{\Psi^t}(\Gamma\otimes I_A)\ket{\Psi^t}.
\label{eq:madv-progress}
\end{equation}
Then \(W_0=1\).

For one active phase query to coordinate \(i\), define
\begin{equation}
R_i(\Gamma)
=
\sup_{v\ne0}
\frac{
\langle v,O_i^*\Gamma O_i v\rangle
}{
\langle v,\Gamma v\rangle
}
=
\norm{\Gamma^{1/2}O_i\Gamma^{-1/2}}^2,
\label{eq:madv-query-ratio}
\end{equation}
and
\begin{equation}
R
=
\max\left\{
1,\max_{i\in[n]}R_i(\Gamma)
\right\}.
\label{eq:madv-R}
\end{equation}
The extra \(1\) accounts for the inactive phase-query sector.  The
one-query part of the multiplicative adversary method gives
\begin{equation}
W_{t+1}\le RW_t,
\qquad
W_T\le R^T.
\label{eq:madv-query-growth}
\end{equation}

\subsection{Final-state bounds}

Equation~\eqref{eq:madv-query-growth} is the part of the framework that is
common to both lower bounds in this paper.  Once it is proved for a chosen
\(\Gamma\), any independent lower bound
\[
W_T\ge L_{\mathrm{fin}}>1
\]
immediately gives
\begin{equation}
T
\ge
\frac{\log L_{\mathrm{fin}}}{\log R}.
\label{eq:madv-direct-final}
\end{equation}
Thus the final-state lower bound need not be obtained from a bad/good-space
threshold.  The unique-OR proof in Section~\ref{sec:uor} uses
\eqref{eq:madv-direct-final}: it derives \(W_T\ge L_{\mathrm{fin}}\)
directly from the success condition.

For approximate counting we instead use the standard bad/good-space
criterion of \cite{Spa07}.  We recall it next.

For \(z\in\{0,1\}\), let
\[
F_z
=
\sum_{x\in X_z}
\ket{x}\!\bra{x}.
\]
Write the spectral decomposition
\[
\Gamma=\sum_\gamma\gamma\Pi_\gamma.
\]
For \(1<\lambda\le\norm{\Gamma}\), define
\[
\Pi_{\mathrm{bad}}
=
\sum_{\gamma<\lambda}\Pi_\gamma,
\qquad
\Pi_{\mathrm{good}}
=
I-\Pi_{\mathrm{bad}}.
\]
The number \(\eta\) below measures how well a state confined to the bad
subspace can already correlate with either correct output. It is not the
bias parameter \(\zeta\). Let \(P_0,P_1\) be the projectors of the final
output measurement of the algorithm.

For the coherent input distribution \(|\delta_x|^2\), define
\[
\operatorname{Succ}_{\delta}(\mathcal A_T)
=
\sum_{x\in X}
|\delta_x|^2
\norm{P_{g(x)}\ket{\psi_x^T}}^2.
\]

\begin{theorem}[Multiplicative adversary bad/good-space bound
{\cite[Theorem~3 and Corollary~4]{Spa07}}]
\label{thm:madv}
Assume
\[
\norm{F_z\Pi_{\mathrm{bad}}}^2\le\eta
\qquad(z=0,1).
\]
If
\[
\operatorname{Succ}_{\delta}(\mathcal A_T)\ge\eta+4\xi,
\]
then
\[
W_T\ge\xi^2\lambda.
\]
Together with \eqref{eq:madv-query-growth},
\[
T
\ge
\frac{\log(\xi^2\lambda)}{\log R}
\]
whenever \(\xi^2\lambda>1\).
\end{theorem}

For completeness, we record the final-state estimate used in this theorem.
Let
\[
\widehat\Pi_{\mathrm{bad}}
=
\Pi_{\mathrm{bad}}\otimes I_A,
\qquad
\widehat\Pi_{\mathrm{good}}
=
\Pi_{\mathrm{good}}\otimes I_A,
\]
and put
\begin{equation}
\beta
=
\norm{
\widehat\Pi_{\mathrm{good}}\ket{\Psi^T}
}^2.
\label{eq:madv-beta}
\end{equation}
If the bad component is nonzero, normalize it to
\(\ket{\Psi_{\mathrm{bad}}}\).  Orthogonality of the two components gives
\[
\ket{\Psi^T}
=
\sqrt{1-\beta}\ket{\Psi_{\mathrm{bad}}}
+
\sqrt{\beta}\ket{\Psi_{\mathrm{good}}},
\]
hence
\[
\norm{\Psi^T-\Psi_{\mathrm{bad}}}
\le 2\sqrt{\beta}.
\]
Every normalized state in the bad subspace has success probability at most
\(\eta\), while changing a normalized state by Euclidean distance \(d\)
changes the expectation of any projector by at most \(2d\).  Therefore
\[
\operatorname{Succ}_{\delta}(\mathcal A_T)
\le
\eta+4\sqrt{\beta}.
\]
If the success is at least \(\eta+4\xi\), then
\[
\beta\ge\xi^2.
\]
Every eigenvalue of \(\Gamma\) on the good subspace is at least \(\lambda\),
so
\[
W_T
\ge
\lambda\beta
\ge
\xi^2\lambda.
\]
This proves the final-state part of Theorem~\ref{thm:madv}.

\begin{corollary}[Bias-normalized form]
\label{cor:madv-bias}
Suppose
\[
\norm{F_z\Pi_{\mathrm{bad}}}^2\le\frac12
\qquad(z=0,1),
\]
and the algorithm has worst-case success probability at least
\[
\frac12+\zeta.
\]
Then the coherent average success is also at least \(1/2+\zeta\).
Taking
\[
\eta=\frac12,
\qquad
\xi=\frac{\zeta}{4}
\]
in Theorem~\ref{thm:madv}, for every
\[
\frac{16}{\zeta^2}<\lambda\le\norm{\Gamma}
\]
we obtain
\begin{equation}
T
\ge
\frac{
\log\!\left(\zeta^2\lambda/16\right)
}{
\log R
}.
\label{eq:madv-bias-form}
\end{equation}
\end{corollary}

The two uses of the framework should therefore be kept separate:
Section~\ref{sec:uor} proves a direct lower bound on \(W_T\) from the
success condition and does not use \(\eta\), while the direct counting proof verifies
\(\eta=1/2\) on an explicit bad subspace and then invokes
Corollary~\ref{cor:madv-bias}.

\section{Direct MADV analysis of approximate counting}
\label{sec:counting-direct}

\subsection{Problem setup and Hamming-layer subspaces}
\label{sec:counting-problem}
\label{sec:hamming-layer-subspaces}

We use the subspaces from the symmetric-function eigenspace method of
\cite{ASW06}. Their threshold analysis uses two adjacent Hamming-weight
layers. Here \(T_{j,w}\) and \(S_{j,w}\) are defined on every weight layer
needed to relate \(M\) and \(M+\Delta\).

Let
\begin{align*}
&X_0=\{x\in\{0,1\}^{N}: |x|=M\},\\
&X_1=\{x\in\{0,1\}^{N}: |x|=M+\Delta\},
\end{align*}
and \(X=X_0\cup X_1\). Define
\[
f(x)=
\begin{cases}
0,&x\in X_0,\\
1,&x\in X_1.
\end{cases}
\]
For the technical analysis in this section we assume
\begin{equation}
M+\Delta\le N-M.
\tag{P}
\label{eq:parameter-convention}
\end{equation}
This condition is without loss of generality. If it fails, complement all
input bits and exchange the two output labels. The promised weights become
\[
N-M-\Delta
\qquad\text{and}\qquad
N-M,
\]
with the same gap \(\Delta\). In the phase-oracle model,
\[
O_{\bar x}=-O_x,
\]
so complementation changes only a global phase and preserves the number of
queries. Thus, replacing the lower promised weight \(M\) by
\(N-M-\Delta\) when necessary, we may assume
\eqref{eq:parameter-convention}. Moreover,
\[
(N-M)(M+\Delta)
\]
is invariant under this transformation, so the lower bound below translates
back to the original pair of Hamming weights without changing its value.

\begin{theorem}[Direct multiplicative-adversary lower bound]
\label{thm:counting-direct-lower-bound}
Let \(1\le M<M+\Delta<N\) and \(0<\zeta\le1/2\). Every quantum query
algorithm that distinguishes Hamming weights \(M\) and \(M+\Delta\) with
success probability at least \(1/2+\zeta\) uses
\begin{equation}
Q=
\Omega\left(
\zeta\frac{\sqrt{(N-M)(M+\Delta)}}{\Delta}
\right)
\label{eq:direct-full-bound}
\end{equation}
queries. Equivalently, if \(\Delta=\epsilon M\),
\[
Q=
\Omega\left(
\frac{\zeta}{\epsilon}
\sqrt{\frac{(1+\epsilon)(N-M)}{M}}
\right).
\]
In the regime \(0<\epsilon\le1\) and \(M+\Delta\le N-M\), this simplifies,
up to constant factors, to
\[
\Omega\left(
\frac{\zeta}{\epsilon}\sqrt{\frac{N}{M}}
\right).
\]
\end{theorem}

For \(0\le w\le N\), let
\[
\mathcal H_w
=
\Span\{\ket{x}:x\in\{0,1\}^N,\ |x|=w\}.
\]
For \(J\subseteq[N]\), with
\[
|J|=j\le\min\{w,N-w\},
\]
define
\[
\ket{\psi_{J,w}}
=
\frac{1}{\sqrt{\binom{N-j}{w-j}}}
\sum_{\substack{x\in\{0,1\}^N\\ |x|=w\\ x_i=1\ \forall i\in J}}
\ket{x}.
\]
Let
\[
T_{j,w}
=
\Span\{\ket{\psi_{J,w}}:J\subseteq[N],\ |J|=j\}.
\]
Note that \(T_{j,w}\subseteq T_{j',w}\) for \(j<j'\), since each
\(\ket{\psi_{J,w}}\) with \(|J|=j\) can be written as a linear combination
of the vectors \(\ket{\psi_{J',w}}\) over extensions \(J'\supseteq J\) with
\(|J'|=j'\).

Define
\[
S_{j,w}
=
T_{j,w}
\cap
T_{j-1,w}^{\perp},
\qquad
T_{-1,w}=\{0\}.
\]
Thus \(S_{j,w}\perp S_{j',w}\) whenever \(j\neq j'\).

Write
\[
\ket{\tilde\psi_{J,w}}
=
\Pi_{T_{j-1,w}^{\perp}}
\ket{\psi_{J,w}},
\qquad
\ket{\ddot\psi_{J,w}}
=
\frac{
\ket{\tilde\psi_{J,w}}
}{
\norm{\ket{\tilde\psi_{J,w}}}
}.
\]
The vectors
\(\ket{\ddot\psi_{J,w}}\) span \(S_{j,w}\), but for \(j>0\) they are
generally not orthogonal.

For the two promise weights, set
\[
w_0=M,
\qquad
w_1=M+\Delta.
\]
For every \(0\le j\le M\), define
\[
S_{j,\pm}
=
\Span\left\{
\frac{
\ket{\ddot\psi_{J,w_0}}
\pm
\ket{\ddot\psi_{J,w_1}}
}{\sqrt2}
:
J\subseteq[N],\ |J|=j
\right\}.
\]
The unitary relation between the two Hamming-weight layers, and the resulting
orthogonal structure of these spaces, is proved below in
Section~\ref{sec:layer-identification}.

\subsection{Adversary matrix and unitary identification}
\label{sec:adversary-matrix}

For adjacent Hamming weights, Claim~16 of \cite{ASW06} gives the unitary
relation used in the threshold analysis. We compose those adjacent maps to
relate the two nonadjacent Hamming-weight layers used here. The spaces
\(S_{j,\pm}\) have already been defined above; the argument below
establishes the unitary relation that gives them the required orthogonal
structure.

Fix \(\theta=1/4\), and set
\[
L=\max\{1,\lfloor\theta M\rfloor\}.
\]
Let \(q>1\) and
\[
\lambda=q^L.
\]
Let \(I_X\) denote the identity on the promise-input space
\[
\mathcal H_X=\mathcal H_M\oplus\mathcal H_{M+\Delta}.
\]
We define
\begin{equation}
\Gamma
=
\sum_{j=0}^{L-1}q^j\Pi_{S_{j,+}}
+
q^L\left(
I_X-\sum_{j=0}^{L-1}\Pi_{S_{j,+}}
\right).
\label{eq:Gamma}
\end{equation}
Thus
\begin{equation}
\mathcal H_{\mathrm{bad}}
=
\bigoplus_{j=0}^{L-1}S_{j,+},
\qquad
\Pi_{\mathrm{bad}}
=
\sum_{j=0}^{L-1}\Pi_{S_{j,+}}.
\label{eq:bad-space}
\end{equation}
The bad projector is independent of \(q\) and \(\lambda\). Once the
orthogonality of the \(S_{j,+}\) spaces is established below, it follows
that \(\Gamma\succ0\), its smallest eigenvalue is \(1\) on \(S_{0,+}\),
and it acts as \(q^L\) on the orthogonal complement of
\(\mathcal H_{\mathrm{bad}}\).

\label{sec:layer-identification}

The next lemma is the adjacent-layer structural input taken from
Claim~16 of \cite{ASW06}; the following theorem is the extension needed here
for a gap of arbitrary size.

\begin{lemma}[Adjacent-layer unitary]
\label{lem:adjacent-layer}
Fix integers
\[
1\le w\le N,
\qquad
0\le j\le\min\{w-1,N-w\}.
\]
There is a unitary
\[
U'_{j,w}:
S_{j,w}
\longrightarrow
S_{j,w-1}
\]
such that, for every \(J\subseteq[N]\) with \(|J|=j\),
\[
U'_{j,w}
\ket{\ddot\psi_{J,w}}
=
\ket{\ddot\psi_{J,w-1}}.
\]
\end{lemma}

\begin{proof}
This is Claim~16 of \cite{ASW06} in the present notation.  Their
normalized \(01\) map sends each normalized projected defining vector at
weight \(w\) to the defining vector with the same set \(J\) at weight
\(w-1\), and the claim proves that this map is unitary.
\end{proof}

\begin{theorem}[Unitary identification between Hamming-weight layers]
\label{thm:layer-identification}
The adjacent-layer maps can be composed to relate arbitrary weights
\(w<w'\).
Fix integers
\[
0\le j\le w<w'\le N-j.
\]
There is a unitary
\[
V_j^{w\to w'}:
S_{j,w}
\longrightarrow
S_{j,w'}
\]
such that, for every \(J\subseteq[N]\) with \(|J|=j\),
\[
V_j^{w\to w'}
\ket{\ddot\psi_{J,w}}
=
\ket{\ddot\psi_{J,w'}}.
\]
Equivalently,
\[
\left\langle
\ddot\psi_{J,w}
\middle|
\ddot\psi_{K,w}
\right\rangle
=
\left\langle
\ddot\psi_{J,w'}
\middle|
\ddot\psi_{K,w'}
\right\rangle
\]
for all \(J,K\subseteq[N]\) of size \(j\).
\end{theorem}

\begin{proof}
For adjacent weights \(r\) and \(r+1\), where
\[
j\le r<N-j,
\]
Lemma~\ref{lem:adjacent-layer} gives a unitary
\[
U'_{j,r+1}:
S_{j,r+1}
\longrightarrow
S_{j,r}
\]
satisfying
\[
U'_{j,r+1}
\ket{\ddot\psi_{J,r+1}}
=
\ket{\ddot\psi_{J,r}}.
\]
Define
\[
V_j^{r\to r+1}
=
\left(U'_{j,r+1}\right)^*.
\]
Then
\[
V_j^{r\to r+1}
\ket{\ddot\psi_{J,r}}
=
\ket{\ddot\psi_{J,r+1}}.
\]

For arbitrary \(w<w'\), define
\[
V_j^{w\to w'}
=
V_j^{w'-1\to w'}
V_j^{w'-2\to w'-1}
\cdots
V_j^{w\to w+1}.
\]
This is a composition of unitaries, and repeated application of the
adjacent identity gives the stated action on every defining vector.  The
Gram-matrix identity follows because \(V_j^{w\to w'}\) preserves inner
products.
\end{proof}

For \(w>w'\), we use the reverse-direction notation
\begin{equation}
V_j^{w\to w'}:=\left(V_j^{w'\to w}\right)^\dagger,
\qquad
V_j^{w\to w}:=I_{S_{j,w}}.
\label{eq:reverse-layer-map}
\end{equation}
Thus \(V_j^{w\to w'}\) denotes a unitary between any two admissible
Hamming-weight subspaces, in either direction, including the identity when
the two weights coincide.

\subsection{Bad-subspace success bound}
\label{sec:bad-success}

The corresponding construction in \cite{ASW06} gives a \(1/2\)-type
overlap between the symmetric subspace and either threshold layer. For the
two nonadjacent Hamming-weight layers used here, the unitary identification
above gives the following exact operator identity.

\begin{theorem}[Bad-subspace success identity]
\label{thm:bad-success}
For \(z=0,1\),
\[
\Pi_{\mathrm{bad}}F_z\Pi_{\mathrm{bad}}
=
\frac12\Pi_{\mathrm{bad}}.
\]
Consequently,
\[
\norm{F_z\Pi_{\mathrm{bad}}}^2
=
\frac12.
\]
\end{theorem}

\begin{proof}
Fix \(j<L\).  By Theorem~\ref{thm:layer-identification}, there is a unitary
\[
V_j
=
V_j^{w_0\to w_1}:
S_{j,w_0}
\longrightarrow
S_{j,w_1}
\]
such that
\[
V_j\ket{\ddot\psi_{J,w_0}}
=
\ket{\ddot\psi_{J,w_1}}.
\]
Choose an orthonormal basis
\[
\{\ket{\phi_{j,\ell}}\}_{\ell=1}^{d_j}
\]
of \(S_{j,w_0}\), and define
\[
\ket{\phi_{j,\ell}^{\pm}}
=
\frac{
\ket{\phi_{j,\ell}}
\pm
V_j\ket{\phi_{j,\ell}}
}{\sqrt2}.
\]
These are orthonormal bases of \(S_{j,+}\) and \(S_{j,-}\), respectively.

The first term of
\(\ket{\phi_{j,\ell}^{+}}\) lies in the weight-\(w_0\) layer, and the
second lies in the weight-\(w_1\) layer.  Hence
\[
F_0\ket{\phi_{j,\ell}^{+}}
=
\frac1{\sqrt2}\ket{\phi_{j,\ell}},
\qquad
F_1\ket{\phi_{j,\ell}^{+}}
=
\frac1{\sqrt2}V_j\ket{\phi_{j,\ell}}.
\]
Because both
\(\{\ket{\phi_{j,\ell}}\}_{\ell}\) and
\(\{V_j\ket{\phi_{j,\ell}}\}_{\ell}\) are orthonormal, for every
\(z\in\{0,1\}\),
\[
\bra{\phi_{j,m}^{+}}
F_z
\ket{\phi_{j,\ell}^{+}}
=
\frac12\delta_{m\ell}.
\]
Expanding the two projectors in the orthonormal basis
\(\{\ket{\phi_{j,\ell}^{+}}\}_{\ell}\) gives
\begin{equation}
\begin{aligned}
\Pi_{S_{j,+}}F_z\Pi_{S_{j,+}}
&=
\sum_{m,\ell=1}^{d_j}
\ket{\phi_{j,m}^{+}}
\bra{\phi_{j,m}^{+}}
F_z
\ket{\phi_{j,\ell}^{+}}
\bra{\phi_{j,\ell}^{+}}\\
&=
\frac12
\sum_{\ell=1}^{d_j}
\ket{\phi_{j,\ell}^{+}}
\bra{\phi_{j,\ell}^{+}}\\
&=
\frac12\Pi_{S_{j,+}}.
\end{aligned}
\label{eq:bad-level}
\end{equation}

We now expand the full bad-space compression:
\[
\Pi_{\mathrm{bad}}F_z\Pi_{\mathrm{bad}}
=
\sum_{j,k=0}^{L-1}
\Pi_{S_{j,+}}F_z\Pi_{S_{k,+}}.
\]
If \(j\ne k\), then
\[
F_zS_{j,+}\subseteq S_{j,w_z},
\qquad
F_zS_{k,+}\subseteq S_{k,w_z}.
\]
Since
\[
S_{j,w_z}\perp S_{k,w_z},
\]
for \(u\in S_{j,+}\) and \(v\in S_{k,+}\), the projector identity
\(F_z^2=F_z\) therefore gives
\[
\langle u,F_zv\rangle
=
\langle F_zu,F_zv\rangle
=
0.
\]
Hence
\[
\Pi_{S_{j,+}}F_z\Pi_{S_{k,+}}=0
\qquad
(j\ne k).
\]
Only the diagonal terms remain, and \eqref{eq:bad-level} yields
\[
\begin{aligned}
\Pi_{\mathrm{bad}}F_z\Pi_{\mathrm{bad}}
&=
\sum_{j=0}^{L-1}
\Pi_{S_{j,+}}F_z\Pi_{S_{j,+}}\\
&=
\frac12
\sum_{j=0}^{L-1}
\Pi_{S_{j,+}}\\
&=
\frac12\Pi_{\mathrm{bad}}.
\end{aligned}
\]

Finally, using \(F_z^*=F_z=F_z^2\),
\[
\begin{aligned}
\norm{F_z\Pi_{\mathrm{bad}}}^2
&=
\norm{
(F_z\Pi_{\mathrm{bad}})^*
(F_z\Pi_{\mathrm{bad}})
}\\
&=
\norm{
\Pi_{\mathrm{bad}}F_z^*F_z\Pi_{\mathrm{bad}}
}\\
&=
\norm{
\Pi_{\mathrm{bad}}F_z\Pi_{\mathrm{bad}}
}\\
&=
\norm{\tfrac12\Pi_{\mathrm{bad}}}\\
&=
\frac12.
\end{aligned}
\]
\end{proof}

Thus Corollary~\ref{cor:madv-bias} applies to the adversary matrix below.

\subsection{Fixed-coordinate block diagonalization and the \(4\times4\) query block}
\label{sec:block-diagonalization}

The starting point is the fixed-coordinate decomposition in Appendix~A of
\cite{ASW06}. Claim~15 decomposes a Hamming-layer subspace according to the
value of the queried bit, and Claim~16 supplies the adjacent-layer unitary
maps. Because our promise layers are separated by \(\Delta\), we use
Theorem~\ref{thm:layer-identification} to identify the four fixed-bit subspaces
\(S_{j,w_a,b}\) with a common reference space and transport one orthonormal
basis to all four. This gives mutually orthogonal four-dimensional blocks on
which both the query and the adversary matrix can be analyzed explicitly.

By permutation symmetry, it suffices to analyze a query to coordinate \(1\),
where
\[
O_1\ket{x}=(-1)^{x_1}\ket{x}.
\]

Fix \(0\le j<L\).  For \(a,b\in\{0,1\}\) and
\(J\subseteq\{2,\ldots,N\}\) with \(|J|=j\), define
\[
\ket{\psi_{J,w_a,b}}
=
\frac{1}{\sqrt{\binom{N-1-j}{w_a-b-j}}}
\sum_{\substack{
x\in\{0,1\}^N\\
|x|=w_a,\ x_1=b,\\
x_i=1\ \forall i\in J
}}
\ket{x}.
\]
Let
\[
T_{j,w_a,b}
=
\Span\left\{
\ket{\psi_{J,w_a,b}}:
J\subseteq\{2,\ldots,N\},\ |J|=j
\right\},
\]
and
\[
S_{j,w_a,b}
=
T_{j,w_a,b}\cap T_{j-1,w_a,b}^{\perp},
\qquad
T_{-1,w_a,b}=\{0\}.
\]
Also define
\[
\ket{\widetilde\psi_{J,w_a,b}}
=
\Pi_{T_{j-1,w_a,b}^{\perp}}
\ket{\psi_{J,w_a,b}},
\]
and
\[
\ket{\ddot\psi_{J,w_a,b}}
=
\frac{
\ket{\widetilde\psi_{J,w_a,b}}
}{
\norm{\ket{\widetilde\psi_{J,w_a,b}}}
}.
\]
Note that
\[
\ket{\ddot\psi_{J,w_a,b}}
=
\ket{b}\otimes
\ket{\ddot\psi^{(N-1)}_{J,w_a-b}}.
\]

These are the same Hamming-layer subspaces as before, restricted to inputs with
\(x_1=b\).  Equivalently, among coordinates \(2,\ldots,N\), the required
Hamming weight is \(w_a-b\).  Applying Theorem~\ref{thm:layer-identification} on coordinates
\(2,\ldots,N\), with \(N\) replaced by \(N-1\), and using the
reverse-direction convention \eqref{eq:reverse-layer-map} when necessary,
define the unitary map between the indicated subspaces by
\[
V_j^{a,b}
:=
\bigl(\ket{b}\bra{0}\bigr)
\otimes V_j^{w_0\to w_a-b}:
S_{j,w_0,0}\longrightarrow S_{j,w_a,b}.
\]
Then
\[
V_j^{a,b}
\ket{\ddot\psi_{J,w_0,0}}
=
\ket{\ddot\psi_{J,w_a,b}}
\]
for every \(J\subseteq\{2,\ldots,N\}\) of size \(j\).  We take
\(V_j^{0,0}=I\).

Let
\[
m_j=\dim S_{j,w_0,0}.
\]
Choose an orthonormal basis
\[
\left\{
\ket{e_{j,\ell}}
\right\}_{\ell=1}^{m_j}
\]
of \(S_{j,w_0,0}\), and define
\[
\ket{\phi_{j,\ell}^{a,b}}
=
V_j^{a,b}\ket{e_{j,\ell}}.
\]
Thus, for each fixed \(a,b\),
\[
\left\{
\ket{\phi_{j,\ell}^{a,b}}
\right\}_{\ell=1}^{m_j}
\]
is an orthonormal basis of \(S_{j,w_a,b}\), and the same label \(\ell\)
refers to the vectors matched by these unitary identifications.

For fixed \((j,\ell)\), define
\begin{equation}
\mathcal B_{j,\ell}
=
\Span\left\{
\ket{\phi_{j,\ell}^{0,0}},
\ket{\phi_{j,\ell}^{0,1}},
\ket{\phi_{j,\ell}^{1,0}},
\ket{\phi_{j,\ell}^{1,1}}
\right\}.
\label{eq:block-space}
\end{equation}
These four fixed-bit basis vectors are orthonormal: different \(a\) belong to different
promise-weight layers, different \(b\) have different values of \(x_1\),
and different \(\ell\) are orthogonal within each \(S_{j,w_a,b}\).
Moreover, the spaces \(S_{j,w_a,b}\) for different \(j\) are orthogonal.
Therefore the spaces \(\mathcal B_{j,\ell}\) are mutually orthogonal.

In the ordered basis
\[
\ket{0,0},\quad
\ket{0,1},\quad
\ket{1,0},\quad
\ket{1,1},
\]
where the first coordinate is \(a\) and the second is \(b=x_1\), the
query is
\begin{equation}
Z=\diag(1,-1,1,-1).
\label{eq:Z}
\end{equation}
The following lemma gives the change of basis from these fixed-bit vectors
to the corresponding vectors in \(S_{j,w_a}\) and \(S_{j+1,w_a}\).
Let
\begin{equation}
\alpha_a^2=\frac{N-w_a-j}{N-2j},
\qquad
\beta_a^2=\frac{w_a-j}{N-2j}.
\label{eq:alpha-beta}
\end{equation}

\begin{lemma}[Fixed-coordinate coefficients]
\label{lem:fixed-coordinate}
For each \(a\in\{0,1\}\), the two orthonormal combinations
\[
\alpha_a\ket{\phi_{j,\ell}^{a,0}}
+
\beta_a\ket{\phi_{j,\ell}^{a,1}}
\]
and
\[
\beta_a\ket{\phi_{j,\ell}^{a,0}}
-
\alpha_a\ket{\phi_{j,\ell}^{a,1}}
\]
lie in \(S_{j,w_a}\) and \(S_{j+1,w_a}\), respectively.  The coefficients are
given by \eqref{eq:alpha-beta}.
\end{lemma}

\begin{proof}
This is the fixed-coordinate decomposition used in the threshold analysis of
\cite{ASW06,Spa07}.  In particular, the coefficient calculation is the
single-instance specialization of Claim~15 of \cite{ASW06}.
The two fixed-bit subspaces are orthogonal because they have different values
of \(x_1\), so the change of basis has the form
\[
\begin{pmatrix}
\alpha_a&\beta_a\\
\beta_a&-\alpha_a
\end{pmatrix}.
\]
The squared norm of the \(x_1=1\) component of a degree-\(j\) vector is
\[
\beta_a^2=\frac{w_a-j}{N-2j},
\]
which is exactly the coefficient from Claim~15 of
\cite{ASW06}; equivalently, it follows by comparing the two fixed-coordinate components after projection away from the lower space.  Therefore
\[
\alpha_a^2
=
1-\beta_a^2
=
\frac{N-w_a-j}{N-2j}.
\]
The orthogonal combination belongs to \(S_{j+1,w_a}\).
\end{proof}

Using the unitary identifications between the two promise layers, the
\(\Gamma\)-eigenbasis in \(\mathcal B_{j,\ell}\) is ordered as
\[
S_{j,+},\quad
S_{j+1,+},\quad
S_{j,-},\quad
S_{j+1,-}.
\]
The change-of-basis matrix from this eigenbasis to the fixed-bit basis is
\begin{equation}
U=
\frac1{\sqrt2}
\begin{pmatrix}
\alpha_0&\beta_0&\alpha_0&\beta_0\\
\beta_0&-\alpha_0&\beta_0&-\alpha_0\\
\alpha_1&\beta_1&-\alpha_1&-\beta_1\\
\beta_1&-\alpha_1&-\beta_1&\alpha_1
\end{pmatrix}.
\label{eq:U}
\end{equation}
The matrix is real orthogonal.  In the \(\Gamma\)-eigenbasis, the query $O_1$ is
\begin{equation}
W=U^*ZU.
\label{eq:W-def}
\end{equation}

Let
\[
\mathcal B_{<L}
=
\bigoplus_{j=0}^{L-1}\bigoplus_{\ell=1}^{m_j}\mathcal B_{j,\ell},
\qquad
\mathcal R=\mathcal B_{<L}^{\perp}.
\]
Each \(\mathcal B_{j,\ell}\) is invariant under \(O_1\), since \(O_1\) is
\(\operatorname{diag}(1,-1,1,-1)\) in its fixed-bit basis. Because
\(O_1=O_1^*\), the orthogonal complement \(\mathcal R\) is also
invariant under \(O_1\).

For every \(k<L\), the standard fixed-bit decomposition of
\(S_{k,w_a}\) separates the defining sets \(J\) into the cases
\(1\notin J\) and \(1\in J\). The first case occurs in the
\(j=k\) blocks and the second in the \(j=k-1\) blocks. Hence
\[
S_{k,+}
\subseteq
\left(\bigoplus_{\ell}\mathcal B_{k,\ell}\right)
\oplus
\left(\bigoplus_{\ell}\mathcal B_{k-1,\ell}\right)
\subseteq \mathcal B_{<L},
\]
where the second summand is omitted for \(k=0\).

Thus \(\mathcal R\) is orthogonal to every eigenspace \(S_{k,+}\)
of \(\Gamma\) with eigenvalue \(q^k<q^L\). By the definition of
\(\Gamma\), every vector \(\ket{w}\in\mathcal R\) is a \(q^L\)-eigenvector of
\(\Gamma\); that is,
\[
\Gamma\ket{w}=q^L\ket{w}.
\]
Therefore, for every \(\ket w\in\mathcal R\),
\[
O_1^*\Gamma O_1 \ket w
=
 O_1^* q^L (O_1\ket w)
=
q^L\ket w
=
\Gamma\ket w,
\]
where in the first equality we used that $\mathcal{R}$ is invariant under $O_1$. Hence the one-query ratio on \(\mathcal R\) is \(1\), and only
the blocks \(B_{j,\ell}\) with \(0\le j<L\) require analysis.

For \(\Delta=1\), the existence of an analogous four-dimensional
decomposition is part of the threshold analysis in Appendix~A.4 of
\cite{ASW06}.  The formulas below are the corresponding calculation for two nonadjacent
Hamming weights; the off-diagonal coefficients depend explicitly on the gap
\(\Delta\).

\begin{theorem}[Direct \(4\times4\) query calculation]
\label{thm:wcalc}
Let \(D_j=N-2j\).  Then
\begin{equation}
W=
\begin{pmatrix}
\rho&\sigma&d&e\\
\sigma&-\rho&e&-d\\
d&e&\rho&\sigma\\
e&-d&\sigma&-\rho
\end{pmatrix},
\label{eq:W}
\end{equation}
where
\begin{equation}
\rho=\frac{N-2M-\Delta}{D_j},
\qquad
d=\frac{\Delta}{D_j},
\end{equation}
\begin{equation}
\begin{aligned}
\sigma
=
\frac{1}{D_j}\Bigl[
&\sqrt{(N-M-j)(M-j)}
\\
&+
\sqrt{(N-M-\Delta-j)(M+\Delta-j)}
\Bigr].
\end{aligned}
\label{eq:sigma}
\end{equation}
and
\begin{equation}
\begin{aligned}
e=
\frac{1}{D_j}\Bigl[&
\sqrt{(N-M-j)(M-j)}\\
&-
\sqrt{(N-M-\Delta-j)(M+\Delta-j)}
].
\label{eq:e}
\end{aligned}
\end{equation}
\end{theorem}

\begin{proof}
Multiplying \eqref{eq:U} and \eqref{eq:Z} directly gives
\[
U^*ZU
=
\begin{pmatrix}
\rho&\sigma&d&e\\
\sigma&-\rho&e&-d\\
d&e&\rho&\sigma\\
e&-d&\sigma&-\rho
\end{pmatrix},
\]
with
\[
\rho=
\frac{
\alpha_0^2-\beta_0^2+
\alpha_1^2-\beta_1^2
}{2},
\]
\[
\sigma=\alpha_0\beta_0+\alpha_1\beta_1,
\]
\begin{equation}
d=
\frac{
\alpha_0^2-\alpha_1^2-
\beta_0^2+\beta_1^2
}{2},
\qquad
e=\alpha_0\beta_0-\alpha_1\beta_1.
\label{eq:rho-sigma-de}
\end{equation}
Substituting
\[
\alpha_a^2=\frac{N-w_a-j}{N-2j},
\qquad
\beta_a^2=\frac{w_a-j}{N-2j},
\]
with \(w_0=M\) and \(w_1=M+\Delta\), yields
\[
\rho=\frac{N-2M-\Delta}{N-2j},
\qquad
d=\frac{\Delta}{N-2j}.
\]
Also,
\[
\alpha_a\beta_a
=
\frac{\sqrt{(N-w_a-j)(w_a-j)}}{N-2j},
\]
which gives \eqref{eq:sigma} and \eqref{eq:e}.
\end{proof}

\subsection{One-query block estimate}
\label{sec:blockratio}

In a block with \(0\le j<L\), the eigenvalues of \(\Gamma\) are
\[
q^j,\quad q^{j+1},\quad q^L,\quad q^L.
\]
After dividing by the common factor \(q^j\), set
\begin{equation}
G_j=\diag(1,q,r_j,r_j),
\qquad
r_j=q^{L-j}.
\label{eq:Gj}
\end{equation}

\begin{theorem}[One-query progress ratio on a \(4\times4\) block]
\label{thm:progress-ratio}
The one-query progress ratio on this block is
\begin{equation}
\sup_{v\ne0}
\frac{\langle v,W^*G_jWv\rangle}{\langle v,G_jv\rangle}
=
\norm{K_j}^2,
\qquad
K_j=G_j^{1/2}WG_j^{-1/2}.
\label{eq:block-ratio}
\end{equation}
\end{theorem}

\begin{proof}
Let \(u=G_j^{1/2}v\).  Then
\[
\frac{\langle v,W^*G_jWv\rangle}{\langle v,G_jv\rangle}
=
\frac{
\langle u,
G_j^{-1/2}W^*G_jWG_j^{-1/2}u
\rangle
}{
\langle u,u\rangle
}.
\]
The matrix in the numerator is \(K_j^*K_j\), so the supremum is
\[
\lambda_{\max}(K_j^*K_j)=\norm{K_j}^2.
\]
\end{proof}

\begin{theorem}[Norm estimate for the \(4\times4\) block]
\label{thm:normestimate}
For the block above,
\begin{equation}
\norm{K_j}
\le
1+
(\sqrt q-1)|\sigma|
+
(\sqrt{r_j}-1)\sqrt{d^2+e^2}.
\label{eq:norm-estimate}
\end{equation}
\end{theorem}

\begin{proof}
Split the ordered basis
\[
S_{j,+},\quad S_{j+1,+},\quad S_{j,-},\quad S_{j+1,-}
\]
into its first two and last two coordinates, and write
\[
A=
\begin{pmatrix}
\rho&\sigma\\
\sigma&-\rho
\end{pmatrix},
\qquad
C=
\begin{pmatrix}
d&e\\
e&-d
\end{pmatrix},
\qquad
D=\diag(1,q).
\]
Then
\[
W=
\begin{pmatrix}
A&C\\
C&A
\end{pmatrix},
\qquad
G_j=
\begin{pmatrix}
D&0\\
0&r_jI_2
\end{pmatrix}.
\]
Consequently,
\begin{equation}
K_j=G_j^{1/2}WG_j^{-1/2}
=
\begin{pmatrix}
D^{1/2}AD^{-1/2}
&
r_j^{-1/2}D^{1/2}C
\\[0.4em]
r_j^{1/2}CD^{-1/2}
&
A
\end{pmatrix}.
\label{eq:Kj-blocks}
\end{equation}

For notation, define
\[
B_{++}
=
D^{1/2}AD^{-1/2}-A,
\]
\[
B_{+-}
=
r_j^{-1/2}D^{1/2}C-C,
\qquad
B_{-+}
=
r_j^{1/2}CD^{-1/2}-C.
\]
The subscripts only indicate the row and column sectors of the corresponding
\(2\times2\) block.  Equation~\eqref{eq:Kj-blocks} is therefore exactly
\begin{equation}
K_j
=
W+
\begin{pmatrix}
B_{++}&B_{+-}\\
B_{-+}&0
\end{pmatrix}.
\label{eq:Kj-block-decomposition}
\end{equation}
Since \(W=U^*ZU\) is unitary, using the triangle inequality for the operator norm, we have 
\begin{equation}
\norm{K_j}
\le
1+\norm{B_{++}}
+
\norm{
\begin{pmatrix}
0&B_{+-}\\
B_{-+}&0
\end{pmatrix}
}.
\label{eq:Kj-block-triangle}
\end{equation}

Directly from the definition,
\[
B_{++}
=
\begin{pmatrix}
0&(\frac1{\sqrt q}-1)\sigma\\
(\sqrt q-1)\sigma&0
\end{pmatrix}.
\]
For a matrix
\[
\begin{pmatrix}0&a\\b&0\end{pmatrix},
\]
the spectral norm is \(\max\{|a|,|b|\}\).  Hence
\begin{equation}
\norm{B_{++}}
=(\sqrt q-1)|\sigma|.
\label{eq:Bpp}
\end{equation}

Moreover,
\begin{equation}
C^2=(d^2+e^2)I_2,
\qquad
\norm{C}=\sqrt{d^2+e^2}.
\label{eq:Cnorm}
\end{equation}
Using these definitions,
\[
B_{-+}
=
C
\begin{pmatrix}
\sqrt{r_j}-1&0\\
0&\sqrt{r_j/q}-1
\end{pmatrix}.
\]
Since \(r_j\ge q\ge1\),
\begin{equation}
\norm{B_{-+}}
\le
(\sqrt{r_j}-1)\sqrt{d^2+e^2}.
\label{eq:Bmplus}
\end{equation}
Similarly,
\[
B_{+-}
=
\begin{pmatrix}
\frac1{\sqrt{r_j}}-1&0\\
0&\sqrt{q/r_j}-1
\end{pmatrix}C.
\]
Because
\[
\left|1-\frac1{\sqrt{r_j}}\right|
\le\sqrt{r_j}-1,
\qquad
\left|1-\sqrt{\frac q{r_j}}\right|
\le\sqrt{r_j}-1,
\]
we also have
\begin{equation}
\norm{B_{+-}}
\le
(\sqrt{r_j}-1)\sqrt{d^2+e^2}.
\label{eq:Bplusm}
\end{equation}
Finally,
\[
\norm{
\begin{pmatrix}
0&B_{+-}\\
B_{-+}&0
\end{pmatrix}
}
=
\max\{\norm{B_{+-}},\norm{B_{-+}}\}.
\]
Substituting \eqref{eq:Bpp}, \eqref{eq:Bmplus},
and \eqref{eq:Bplusm} into \eqref{eq:Kj-block-triangle} proves
\eqref{eq:norm-estimate}.
\end{proof}

\subsubsection{Coefficient estimates}
\label{sec:coefficients}

The exact \(4\times4\) coefficients retain the dependence on both distances
to the ends of the Hamming cube. This lets us avoid the additional
restriction \(\Delta\le M\).

\begin{theorem}[Coefficient estimates]
\label{thm:coefficients}
Uniformly for \(0\le j<L\),
\begin{equation}
\begin{aligned}
|\sigma|
&=
O\left(
\sqrt{\frac{M+\Delta}{N-M}}
\right),\\
\sqrt{d^2+e^2}
&=
O\left(
\frac{\Delta}{\sqrt{(N-M)(M+\Delta)}}
\right).
\end{aligned}
\label{eq:coefficient-estimates}
\end{equation}
All implicit constants are absolute.
\end{theorem}

\begin{proof}
Set
\[
A=M-j,
\qquad
B=N-M-j.
\]
Then \(N-2j=A+B\), and the convention
\(M+\Delta\le N-M\) implies
\[
A+\Delta\le B,
\qquad
B-\Delta\ge A.
\]
From \eqref{eq:sigma},
\begin{align*}
|\sigma|
&=
\frac{\sqrt{AB}+\sqrt{(B-\Delta)(A+\Delta)}}{A+B}\\
&\le
\frac{\sqrt{AB}+\sqrt{B(A+\Delta)}}{B}\\
&\le
2\sqrt{\frac{A+\Delta}{B}}.
\end{align*}
For \(j<L\), either \(j=0\) or \(j<M/4\). Since
\(N-M\ge M+\Delta\ge M\), this gives uniformly
\[
A+\Delta=\Theta(M+\Delta),
\qquad
B=\Theta(N-M),
\]
and therefore
\[
|\sigma|
=
O\left(\sqrt{\frac{M+\Delta}{N-M}}\right).
\]

For the off-diagonal coefficients, use
\eqref{eq:rho-sigma-de} and \(\alpha_a^2+\beta_a^2=1\). A direct
calculation gives
\begin{equation}
\sqrt{d^2+e^2}
=
\alpha_0\beta_1-\beta_0\alpha_1.
\label{eq:de-angle-identity}
\end{equation}
The right-hand side is nonnegative because
\(\beta_1\ge\beta_0\) and \(\alpha_0\ge\alpha_1\).
Substituting the coefficients from \eqref{eq:alpha-beta},
\begin{align*}
\sqrt{d^2+e^2}
&=
\frac{\sqrt{B(A+\Delta)}-\sqrt{A(B-\Delta)}}{A+B}\\
&=
\frac{\Delta}{
\sqrt{B(A+\Delta)}+\sqrt{A(B-\Delta)}
}\\
&\le
\frac{\Delta}{\sqrt{B(A+\Delta)}}.
\end{align*}
Using again
\(B=\Theta(N-M)\) and
\(A+\Delta=\Theta(M+\Delta)\) proves the second estimate in
\eqref{eq:coefficient-estimates}.
\end{proof}

Let
\[
A_j
=
(\sqrt q-1)|\sigma|
+
(\sqrt{r_j}-1)\sqrt{d^2+e^2}.
\]
Theorems~\ref{thm:progress-ratio}, \ref{thm:normestimate}, and
\ref{thm:coefficients} give
\[
R_j\le(1+A_j)^2
\]
and
\begin{equation}
\begin{aligned}
A_j
\le C\Biggl[
&(\sqrt q-1)
\sqrt{\frac{M+\Delta}{N-M}}\\
&+
(\sqrt\lambda-1)
\frac{\Delta}{\sqrt{(N-M)(M+\Delta)}}
\Biggr].
\end{aligned}
\label{eq:Aj-pre}
\end{equation}
Since \(q^L=\lambda\),
\[
\sqrt q=(\sqrt\lambda)^{1/L}.
\]
For \(y\ge1\), the factorization
\[
y-1=(y^{1/L}-1)
\left(1+y^{1/L}+\cdots+y^{(L-1)/L}\right)
\]
implies
\begin{equation}
y^{1/L}-1\le\frac{y-1}{L}.
\label{eq:root-bound}
\end{equation}
Applying this with \(y=\sqrt\lambda\), and using
\[
\frac{M}{8}\le L\le M
\qquad(M\ge1),
\]
gives
\[
\sqrt q-1
\le
C\frac{\sqrt\lambda-1}{M}.
\]
Moreover, because \(M,\Delta\ge1\),
\[
M+\Delta\le2M\Delta.
\]
Consequently,
\[
(\sqrt q-1)
\sqrt{\frac{M+\Delta}{N-M}}
\le
C'(\sqrt\lambda-1)
\frac{\Delta}{\sqrt{(N-M)(M+\Delta)}}.
\]
Substituting this into \eqref{eq:Aj-pre},
\begin{equation}
A_j
\le
C''
(\sqrt\lambda-1)
\frac{\Delta}{\sqrt{(N-M)(M+\Delta)}}.
\label{eq:Aj-final}
\end{equation}
Therefore
\begin{equation}
\begin{aligned}
\log R
&\le
2\max_j\log(1+A_j)\\
&\le
2\max_j A_j\\
&\le
C_0
(\sqrt\lambda-1)
\frac{\Delta}{\sqrt{(N-M)(M+\Delta)}}.
\end{aligned}
\label{eq:logR}
\end{equation}

\subsection{Small-bias lower bound}
\label{sec:main-proof}

\begin{proof}[Proof of Theorem~\ref{thm:counting-direct-lower-bound}]
By Theorem~\ref{thm:bad-success}, the hypotheses of
Corollary~\ref{cor:madv-bias} hold. Choose
\[
\lambda=\frac{64}{\zeta^2},
\qquad
q=\lambda^{1/L}.
\]
Then
\[
\log\!\left(\frac{\zeta^2\lambda}{16}\right)=\log4,
\qquad
\sqrt\lambda-1\le\frac8\zeta.
\]
By \eqref{eq:logR},
\[
\log R
\le
C_1
\frac{\Delta}{\zeta\sqrt{(N-M)(M+\Delta)}}.
\]
Hence Corollary~\ref{cor:madv-bias} gives
\[
Q
\ge
\frac{\log4}{\log R}
\ge
c\zeta
\frac{\sqrt{(N-M)(M+\Delta)}}{\Delta}.
\]
This proves \eqref{eq:direct-full-bound} under the convention
\eqref{eq:parameter-convention}. If the original pair does not satisfy that
convention, complementing all bits replaces \(M\) by
\(N-M-\Delta\). The product in the numerator is unchanged:
\[
(N-(N-M-\Delta))(N-M)
=(M+\Delta)(N-M).
\]
Thus the same bound holds for the original pair of Hamming weights.

If \(\Delta=\epsilon M\), then
\[
\frac{\sqrt{(N-M)(M+\Delta)}}{\Delta}
=
\frac1\epsilon
\sqrt{\frac{(1+\epsilon)(N-M)}{M}},
\]
which gives the stated equivalent form.
\end{proof}

\subsection{A bound on the final good-space weight}

\label{sec:good-space-weight}

The following statement concerns the coherent input superposition used in the
adversary argument. Define
\begin{equation}
\ket{u_a}
=
\frac1{\sqrt{|X_a|}}
\sum_{x\in X_a}\ket{x},
\qquad
\ket{\delta}
=
\frac{\ket{u_0}+\ket{u_1}}{\sqrt2}
\in S_{0,+}.
\label{eq:delta}
\end{equation}
For a \(T\)-query algorithm, let \(\ket{\psi_x^T}\) be its final accessible
state on input \(x\), and run the algorithm coherently from \(\ket{\delta}\):
\begin{equation}
\begin{aligned}
\ket{\Psi^T}
=
\frac{1}{\sqrt2}\Biggl(
&\frac{1}{\sqrt{|X_0|}}
\sum_{x\in X_0}\ket{x}\ket{\psi_x^T}
\\
&+
\frac{1}{\sqrt{|X_1|}}
\sum_{x\in X_1}\ket{x}\ket{\psi_x^T}
\Biggr).
\end{aligned}
\label{eq:coherent-final}
\end{equation}
Let
\[
\widehat\Pi_{\mathrm{bad}}
=
\Pi_{\mathrm{bad}}\otimes I_A.
\]

\begin{corollary}[Good-space weight after \(T\) queries]
\label{cor:good-space-weight}
Under the convention \eqref{eq:parameter-convention}, there is an absolute
constant \(C>0\) such that
\begin{equation}
\beta_T
:=
\left\|
(I-\widehat\Pi_{\mathrm{bad}})\ket{\Psi^T}
\right\|^2
\le
C
\frac{T^2\Delta^2}{(N-M)(M+\Delta)}.
\label{eq:good-space-weight}
\end{equation}
If \(\widehat\Pi_{\mathrm{bad}}\ket{\Psi^T}\ne0\) and
\[
\ket{\Psi_{\mathrm{bad}}}
=
\frac{\widehat\Pi_{\mathrm{bad}}\ket{\Psi^T}}
{\|\widehat\Pi_{\mathrm{bad}}\ket{\Psi^T}\|},
\]
then
\[
\|\ket{\Psi^T}-\ket{\Psi_{\mathrm{bad}}}\|^2
\le2\beta_T.
\]
\end{corollary}

\begin{proof}
The projector \(\Pi_{\mathrm{bad}}\) depends only on \(L\), not on \(q\)
or \(\lambda\). If \(\beta_T=0\), the claim is immediate. Otherwise set
\[
\lambda=\frac4{\beta_T},
\qquad
q=\lambda^{1/L}.
\]
The initial state \(\ket\delta\) lies in \(S_{0,+}\), so \(W_0=1\), and the
one-query estimate gives
\[
W_T\le R^T.
\]
Every vector orthogonal to the bad subspace has adversary eigenvalue
\(\lambda\), hence
\[
W_T
=
\bra{\Psi^T}(\Gamma\otimes I_A)\ket{\Psi^T}
\ge
\lambda\beta_T
=4.
\]
Therefore
\[
\log4\le T\log R.
\]
Applying \eqref{eq:logR} with \(\lambda=4/\beta_T\),
\[
\begin{aligned}
\log R
&\le
C_0\left(\frac2{\sqrt{\beta_T}}-1\right)
\frac{\Delta}{\sqrt{(N-M)(M+\Delta)}}\\
&\le
\frac{2C_0}{\sqrt{\beta_T}}
\frac{\Delta}{\sqrt{(N-M)(M+\Delta)}}.
\end{aligned}
\]
Rearranging proves \eqref{eq:good-space-weight}.

For the final statement, decompose
\[
\ket{\Psi^T}
=
\sqrt{1-\beta_T}\ket{\Psi_{\mathrm{bad}}}
+
\sqrt{\beta_T}\ket{\Psi_{\mathrm{good}}},
\]
where the two normalized states are orthogonal. Then
\[
\|\ket{\Psi^T}-\ket{\Psi_{\mathrm{bad}}}\|^2
=
2(1-\sqrt{1-\beta_T})
\le2\beta_T.
\]
\end{proof}

\section{Small-bias unique OR}
\label{sec:uor}

This section gives an independent multiplicative-adversary proof for the
promise \(\{0^n,e_1,\ldots,e_n\}\). The resulting lower bound is used in
the reduction of Section~\ref{sec:uor-to-counting}.

\subsection{Problem and adversary construction}

Consider the promise problem
\[
\operatorname{UOR}_n:\{0^n,e_1,\ldots,e_n\}\to\{0,1\},
\]
where
\[
\operatorname{UOR}_n(0^n)=0,
\qquad
\operatorname{UOR}_n(e_i)=1.
\]
We prove the following.

\begin{theorem}[Small-bias unique OR]
\label{thm:uor-lower-bound}
Let $0<\zeta\le 1/2$. Any quantum query algorithm that computes
$\operatorname{UOR}_n$ with worst-case success probability at least
\[
\frac12+\zeta
\]
uses
\[
T=\Omega(\sqrt{\zeta n})
\]
queries.
\end{theorem}

The proof uses only the promise inputs $0^n,e_1,\ldots,e_n$.
We use the query-growth part of the multiplicative-adversary framework from
Section~\ref{sec:madv-framework}, but we do not introduce a bad/good
threshold for this problem.  Instead, the final value $W_T$ is estimated
directly from the success condition.

Let
\[
\cH_I=\Span\{\ket{0},\ket{1},\ldots,\ket{n}\},
\]
where $\ket{0}$ labels the input $0^n$ and $\ket{i}$ labels $e_i$.
Define the uniform positive-input vector
\[
\ket{u}=\frac1{\sqrt n}\sum_{i=1}^n\ket{i},
\]
and
\[
\ket{+}=\frac{\ket{0}+\ket{u}}{\sqrt2},
\qquad
\ket{-}=\frac{\ket{0}-\ket{u}}{\sqrt2}.
\]
Inside the positive-input space
\[
\cH_1=\Span\{\ket{1},\ldots,\ket{n}\},
\]
define
\[
\cS=
\left\{
\sum_{i=1}^n c_i\ket{i}:
\sum_{i=1}^n c_i=0
\right\}.
\]
For $\ket{v}=\sum_i c_i\ket{i}\in\cH_1$,
\[
\braket{u}{v}=\frac1{\sqrt n}\sum_{i=1}^n c_i,
\]
so $\cS$ is exactly the subspace of $\cH_1$ orthogonal to $\ket{u}$.
Hence
\[
\cH_1=\Span\{\ket{u}\}\oplus\cS,
\]
and therefore
\begin{equation}
\cH_I
=
\Span\{\ket{+}\}
\oplus\cS
\oplus\Span\{\ket{-}\}.
\label{eq:1}
\end{equation}

Let
\[
\Pi_u=\ket{u}\!\bra{u}
\]
denote the projector onto the uniform direction in \(\cH_1\).

Let $\ket{\psi_x^t}$ denote the accessible state of the algorithm after $t$
queries on input $x$. We may assume the computation is unitary until the
final two-outcome measurement.

Define the uniform positive-input state
\begin{equation}
\ket{\Psi_u^t}
=
\frac1{\sqrt n}
\sum_{i=1}^n
\ket{i}\otimes\ket{\psi_i^t}.
\label{eq:2}
\end{equation}
Define $\ket{a^t}$ by projecting $\ket{\Psi_u^t}$ onto the uniform direction:
\begin{equation}
(\Pi_u\otimes I_A)\ket{\Psi_u^t}
=
\ket{u}\otimes\ket{a^t}.
\label{eq:3}
\end{equation}
Equivalently,
\begin{equation}
\ket{a^t}
=
(\bra{u}\otimes I_A)\ket{\Psi_u^t}
=
\frac1n\sum_{i=1}^n\ket{\psi_i^t}.
\label{eq:4}
\end{equation}
Since $\ket{\Psi_u^t}$ is normalized and the projections onto
$\Span\{\ket u\}$ and $\cS$ are orthogonal,
\begin{equation}
\norm{(\Pi_{\cS}\otimes I_A)\ket{\Psi_u^t}}^2
=
1-\norm{a^t}^2.
\label{eq:5}
\end{equation}

Now run the algorithm coherently from
\[
\ket{+}
=
\frac1{\sqrt2}\ket{0}
+
\frac1{\sqrt{2n}}\sum_{i=1}^n\ket{i}.
\]
After $t$ queries the joint state is
\begin{equation}
\ket{\Psi^t}
=
\frac1{\sqrt2}\ket{0}\ket{\psi_0^t}
+
\frac1{\sqrt2}\ket{\Psi_u^t}.
\label{eq:6}
\end{equation}

For a parameter $q>1$, define
\begin{equation}
\Gamma_q
=
\ket{+}\!\bra{+}
+q\Pi_{\cS}
+q^2\ket{-}\!\bra{-}.
\label{eq:7}
\end{equation}
The progress function is
\begin{equation}
W_t
=
\bra{\Psi^t}(\Gamma_q\otimes I_A)\ket{\Psi^t}.
\label{eq:8}
\end{equation}
Initially, $\ket{\Psi^0}=\ket{+}\ket{\psi^0}$, so
\begin{equation}
W_0=1.
\label{eq:9}
\end{equation}
Accessible-memory unitaries do not change $W_t$ because they act trivially on
$\cH_I$.

\subsection{One-query progress}

Fix a queried coordinate $j\in[n]$. On the input-label register, the active
Boolean phase query is
\begin{equation}
O_j\ket{0}=\ket{0},
\qquad
O_j\ket{i}=(-1)^{\delta_{ij}}\ket{i}.
\label{eq:10}
\end{equation}
Equivalently, on $\cH_1$,
\[
O_j=I-2\ket{j}\!\bra{j}.
\]
Define
\begin{equation}
\ket{s_j}
=
\sqrt{\frac{n}{n-1}}
\left(
\ket{j}-\frac1{\sqrt n}\ket{u}
\right).
\label{eq:11}
\end{equation}
Then $\ket{s_j}\in\cS$, $\norm{s_j}=1$, and
\begin{equation}
\ket{j}
=
\frac1{\sqrt n}\ket{u}
+
\sqrt{\frac{n-1}{n}}\ket{s_j}.
\label{eq:12}
\end{equation}
If $\ket{v}\in\cS$ and $\braket{s_j}{v}=0$, then
\[
\braket{j}{v}
=
\frac1{\sqrt n}\braket{u}{v}
+
\sqrt{\frac{n-1}{n}}\braket{s_j}{v}
=0,
\]
so $O_j\ket{v}=\ket{v}$. Hence the only nontrivial block is
\[
\Span\{\ket{+},\ket{s_j},\ket{-}\}.
\]

Using
\[
O_j\ket{u}
=
\ket{u}-\frac2{\sqrt n}\ket{j},
\]
we obtain
\begin{equation}
O_j\ket{+}
=
\left(1-\frac1n\right)\ket{+}
-\frac{\sqrt{2(n-1)}}{n}\ket{s_j}
+\frac1n\ket{-},
\label{eq:13}
\end{equation}
\begin{equation}
\begin{aligned}
O_j\ket{s_j}
=&
-\frac{\sqrt{2(n-1)}}{n}\ket{+}
+\left(-1+\frac2n\right)\ket{s_j}\\
&+\frac{\sqrt{2(n-1)}}{n}\ket{-},
\end{aligned}
\label{eq:14}
\end{equation}
and
\begin{equation}
O_j\ket{-}
=
\frac1n\ket{+}
+\frac{\sqrt{2(n-1)}}{n}\ket{s_j}
+\left(1-\frac1n\right)\ket{-}.
\label{eq:15}
\end{equation}
Thus, in the ordered orthonormal basis
$\{\ket{+},\ket{s_j},\ket{-}\}$,
\begin{equation}
O_j=
\begin{pmatrix}
1-\frac1n & -\frac{\sqrt{2(n-1)}}{n} & \frac1n\\[2mm]
-\frac{\sqrt{2(n-1)}}{n} & -1+\frac2n & \frac{\sqrt{2(n-1)}}{n}\\[2mm]
\frac1n & \frac{\sqrt{2(n-1)}}{n} & 1-\frac1n
\end{pmatrix}.
\label{eq:16}
\end{equation}

The matrix elements connecting
\[
\ket{+}\leftrightarrow\cS,
\qquad
\cS\leftrightarrow\ket{-}
\]
are of order $\Theta(n^{-1/2})$, whereas the matrix element connecting
$\ket{+}$ and $\ket{-}$ is of order $\Theta(n^{-1})$.

On the nontrivial three-dimensional block,
\[
\Gamma_q=\operatorname{diag}(1,q,q^2).
\]
The one-query progress ratio on this block is
\begin{equation}
\sup_{v\neq0}
\frac{\bra{v}O_j^*\Gamma_qO_j\ket{v}}
{\bra{v}\Gamma_q\ket{v}}
=
\opnorm{\Gamma_q^{1/2}O_j\Gamma_q^{-1/2}}^2.
\label{eq:17}
\end{equation}
Since $O_j$ is unitary,
\[
\opnorm{O_j}=1.
\]
A direct entrywise computation from \eqref{eq:16} gives
\begin{equation}
\begin{aligned}
&\opnorm{\Gamma_q^{1/2}O_j\Gamma_q^{-1/2}-O_j}_F^2\\
&=
\frac{4(n-1)}{n^2}
\left[
(\sqrt q-1)^2+(1-q^{-1/2})^2
\right]\\
&\quad+
\frac1{n^2}
\left[
(q-1)^2+(1-q^{-1})^2
\right].
\end{aligned}
\label{eq:18}
\end{equation}
When restricting $1\le q\le n$, we have
\[
(\sqrt q-1)^2+(1-q^{-1/2})^2\le 2q
\]
and
\[
(q-1)^2+(1-q^{-1})^2\le 2q^2.
\]
Therefore
\begin{equation}
\opnorm{\Gamma_q^{1/2}O_j\Gamma_q^{-1/2}-O_j}_F^2
\le
\frac{10q}{n}.
\label{eq:19}
\end{equation}
Hence
\[
\opnorm{\Gamma_q^{1/2}O_j\Gamma_q^{-1/2}}
\le
1+\sqrt{\frac{10q}{n}},
\]
and therefore
\begin{equation}
\sup_{v\neq0}
\frac{\bra{v}O_j^*\Gamma_qO_j\ket{v}}
{\bra{v}\Gamma_q\ket{v}}
\le
\left(1+\sqrt{\frac{10q}{n}}\right)^2.
\label{eq:20}
\end{equation}

For the complete Boolean phase oracle,
\[
O
=
\sum_{j=1}^n\sum_{p\in\{0,1\}}
O_j^p\otimes\ket{j,p}\!\bra{j,p}_Q\otimes I_W.
\]
The $p=0$ sector is the identity, while the $p=1$ sectors obey \eqref{eq:20}.
Because different query-register sectors are orthogonal, the standard
multiplicative-adversary argument gives
\begin{equation}
W_{t+1}
\le
\left(1+\sqrt{\frac{10q}{n}}\right)^2W_t.
\label{eq:21}
\end{equation}
Thus
\begin{equation}
W_T
\le
\left(1+\sqrt{\frac{10q}{n}}\right)^{2T}.
\label{eq:22}
\end{equation}
Using $\log(1+x)\le x$,
\begin{equation}
\log W_T
\le
2T\sqrt{\frac{10q}{n}}.
\label{eq:23}
\end{equation}

\subsection{Final-state bound}

From \eqref{eq:3},
\begin{equation}
\ket{\Psi_u^T}
=
\ket{u}\ket{a^T}
+
(\Pi_{\cS}\otimes I_A)\ket{\Psi_u^T}.
\label{eq:24}
\end{equation}
Substituting this into \eqref{eq:6}, and using
\[
\ket{0}=\frac{\ket{+}+\ket{-}}{\sqrt2},
\qquad
\ket{u}=\frac{\ket{+}-\ket{-}}{\sqrt2},
\]
we get
\begin{equation}
\begin{aligned}
\ket{\Psi^T}
={}&
\frac12\ket{+}
\bigl(\ket{\psi_0^T}+\ket{a^T}\bigr)\\
&+
\frac1{\sqrt2}(\Pi_{\cS}\otimes I_A)\ket{\Psi_u^T}\\
&+
\frac12\ket{-}
\bigl(\ket{\psi_0^T}-\ket{a^T}\bigr).
\end{aligned}
\label{eq:25}
\end{equation}
The three terms are mutually orthogonal because their input-label parts lie in
$\Span\{\ket{+}\}$, $\cS$, and $\Span\{\ket{-}\}$, respectively.
Therefore
\begin{equation}
\begin{aligned}
W_T
={}&
\frac14\norm{\psi_0^T+a^T}^2
+
\frac q2
\norm{(\Pi_{\cS}\otimes I_A)\ket{\Psi_u^T}}^2\\
&+
\frac{q^2}{4}\norm{\psi_0^T-a^T}^2.
\end{aligned}
\label{eq:26}
\end{equation}
Using \eqref{eq:5},
\begin{equation}
W_T
=
\frac14\norm{\psi_0^T+a^T}^2
+
\frac q2(1-\norm{a^T}^2)
+
\frac{q^2}{4}\norm{\psi_0^T-a^T}^2.
\label{eq:27}
\end{equation}
Also,
\begin{equation}
\frac14\norm{\psi_0^T+a^T}^2
+
\frac12(1-\norm{a^T}^2)
+
\frac14\norm{\psi_0^T-a^T}^2
=1.
\label{eq:28}
\end{equation}
Hence
\begin{equation}
W_T
=
1
+
\frac{q-1}{2}(1-\norm{a^T}^2)
+
\frac{q^2-1}{4}\norm{\psi_0^T-a^T}^2.
\label{eq:29}
\end{equation}

Let $P_1$ be the projector corresponding to output $1$.
Worst-case success probability at least $1/2+\zeta$ implies
\[
\bra{\psi_i^T}P_1\ket{\psi_i^T}
\ge
\frac12+\zeta
\qquad(i\in[n]),
\]
and
\[
\bra{\psi_0^T}P_1\ket{\psi_0^T}
\le
\frac12-\zeta.
\]
Averaging over $i$ and subtracting gives
\begin{equation}
2\zeta
\le
\Tr\!\left[
P_1\left(
\frac1n\sum_{i=1}^n
\ket{\psi_i^T}\!\bra{\psi_i^T}
-
\ket{\psi_0^T}\!\bra{\psi_0^T}
\right)
\right].
\label{eq:success-average}
\end{equation}

Expanding the average positive-input density matrix gives
\begin{equation}
\begin{aligned}
\frac1n\sum_{i=1}^n
\ket{\psi_i^T}\!\bra{\psi_i^T}
=
\frac1n\sum_{i=1}^n
&(\ket{\psi_i^T-a^T}+\ket{a^T})\\
&\times(\bra{\psi_i^T-a^T}+\bra{a^T}).
\end{aligned}
\label{eq:decompose-psi}
\end{equation}

Using
\[
\ket{a^T}=\frac1n\sum_{i=1}^n\ket{\psi_i^T},
\]
we have
\[
\sum_{i=1}^n
\bigl(\ket{\psi_i^T}-\ket{a^T}\bigr)=0.
\]
Hence by canceling the cross terms of \eqref{eq:decompose-psi},
\begin{equation}
\begin{aligned}
\frac{1}{n}\sum_{i=1}^{n}
\ket{\psi_i^T}\!\bra{\psi_i^T}
={}&
\ket{a^T}\!\bra{a^T}
\\
&+
\frac{1}{n}\sum_{i=1}^{n}
\bigl(\ket{\psi_i^T}-\ket{a^T}\bigr)
\bigl(\bra{\psi_i^T}-\bra{a^T}\bigr).
\end{aligned}
\label{eq:positive-density-decomposition}
\end{equation}
Substituting into \eqref{eq:success-average},
\begin{align}
2\zeta
\le{}&
\bra{a^T}P_1\ket{a^T}
-
\bra{\psi_0^T}P_1\ket{\psi_0^T}
\nonumber\\
&+
\frac1n\sum_{i=1}^n
\bigl(\bra{\psi_i^T}-\bra{a^T}\bigr)
P_1
\bigl(\ket{\psi_i^T}-\ket{a^T}\bigr).
\label{eq:success-split}
\end{align}
Since $0\preceq P_1\preceq I$,
\[
\frac1n\sum_{i=1}^n
\bigl(\bra{\psi_i^T}-\bra{a^T}\bigr)
P_1
\bigl(\ket{\psi_i^T}-\ket{a^T}\bigr)
\le
\frac1n\sum_{i=1}^n
\norm{\psi_i^T-a^T}^2.
\]
Moreover,
\begin{equation}
\frac1n\sum_{i=1}^n\norm{\psi_i^T-a^T}^2
=
1-\norm{a^T}^2.
\label{eq:variance-identity}
\end{equation}
For the first two terms in \eqref{eq:success-split},
\[
\begin{aligned}
&\bra{a^T}P_1\ket{a^T}
-
\bra{\psi_0^T}P_1\ket{\psi_0^T}\\
&=
(\bra{a^T}-\bra{\psi_0^T})P_1\ket{a^T}
+
\bra{\psi_0^T}P_1(\ket{a^T}-\ket{\psi_0^T}),
\end{aligned}
\]
and therefore
\begin{equation}
\left|
\bra{a^T}P_1\ket{a^T}
-
\bra{\psi_0^T}P_1\ket{\psi_0^T}
\right|
\le
2\norm{\psi_0^T-a^T}.
\label{eq:measurement-difference}
\end{equation}
Combining the preceding estimates,
\begin{equation}
2\zeta
\le
1-\norm{a^T}^2
+
2\norm{\psi_0^T-a^T}.
\label{eq:success-bound}
\end{equation}

Since $\norm{\psi_0^T}=1$ and $\norm{a^T}\le 1$,
\begin{align}
1-\norm{a^T}^2
&=
(1-\norm{a^T})(1+\norm{a^T})\nonumber\\
&\le
2(1-\norm{a^T})\nonumber\\
&\le
2\norm{\psi_0^T-a^T},
\label{eq:norm-loss-bound}
\end{align}
where the last step is the reverse triangle inequality.
Together with \eqref{eq:success-bound}, this yields
\begin{equation}
\norm{\psi_0^T-a^T}
\ge
\frac{\zeta}{2}.
\label{eq:final-displacement}
\end{equation}
Consequently,
\begin{equation}
\norm{\psi_0^T-a^T}^2
\ge
\frac{\zeta^2}{4}.
\label{eq:final-displacement-squared}
\end{equation}

By \eqref{eq:29} and \eqref{eq:final-displacement-squared}, it is enough to
choose $q=\Theta(1/\zeta)$. We set
\begin{equation}
q=\frac4\zeta.
\label{eq:q-choice}
\end{equation}
First suppose
\begin{equation}
\zeta\ge\frac4n.
\label{eq:large-bias-regime}
\end{equation}
Then $q\le n$, so the one-query estimate above applies.

Dropping the nonnegative middle term in \eqref{eq:29} and using
\eqref{eq:final-displacement-squared},
\[
\begin{aligned}
W_T
&\ge
1+
\frac{q^2-1}{4}\norm{\psi_0^T-a^T}^2\\
&\ge
1+
\frac14\left(\frac{16}{\zeta^2}-1\right)
\frac{\zeta^2}{4}\\
&=
2-\frac{\zeta^2}{16}
\ge
\frac{127}{64}.
\end{aligned}
\]
Therefore
\begin{equation}
W_T\ge\frac{127}{64}.
\label{eq:final-state-lower}
\end{equation}
Combining \eqref{eq:23} and \eqref{eq:final-state-lower},
\[
\log\frac{127}{64}
\le
2T\sqrt{\frac{10q}{n}}.
\]
Substituting $q=4/\zeta$ gives
\begin{equation}
T
\ge
\frac{\log(127/64)}{4\sqrt{10}}
\sqrt{\zeta n}.
\label{eq:main-lower-bound}
\end{equation}

Finally, if
\begin{equation}
0<\zeta<\frac4n,
\label{eq:small-bias-regime}
\end{equation}
then a zero-query algorithm cannot have positive bias because its output
distribution is independent of the input. Hence $T\ge1$. Since
\[
\sqrt{\zeta n}<2,
\]
we also have
\[
T\ge\frac12\sqrt{\zeta n}.
\]
Combining the two regimes proves
\[
T=\Omega(\sqrt{\zeta n}).
\]

\section{Unique-OR reduction and final bound}
\label{sec:reduction-final}

\subsection{Reduction from unique OR to approximate counting}
\label{sec:uor-to-counting}

\begin{theorem}[Unique-OR reduction lower bound]
\label{thm:counting-reduction-bound}
Every quantum query algorithm that distinguishes Hamming weights \(M\) and
\(M+\Delta\), with \(1\le M<M+\Delta<N\), and succeeds with probability at
least \(1/2+\zeta\), uses
\begin{equation}
Q
=
\Omega\left(
\sqrt{\frac{\zeta N}{\Delta}}
\right)
\label{eq:counting-reduction-bound}
\end{equation}
queries.
\end{theorem}

\begin{proof}
Complement all input bits and exchange the output labels if necessary so
that
\[
M+\Delta\le N-M.
\]
This does not change \(N\), \(\Delta\), or the query complexity.
Set
\begin{equation}
N'
=
\left\lfloor\frac{N-M}{\Delta}\right\rfloor.
\label{eq:reduction-Nprime}
\end{equation}
For a unique-OR input \(y\in\{0,1\}^{N'}\), define
\begin{equation}
x(y)
=
\underbrace{y\,y\,\cdots\,y}_{\Delta\text{ copies}}
\,0^{\,N-M-\Delta N'}
\,1^M.
\label{eq:uor-counting-map-general}
\end{equation}
If \(y=0^{N'}\), then \(|x(y)|=M\). If \(|y|=1\), each repeated copy
contributes one marked item and \(|x(y)|=M+\Delta\).

A query to \(x(y)\) can be simulated with at most one query to \(y\): a
position in a repeated block determines one coordinate of \(y\), while the
padding zeros and final ones are known. Hence a \(Q\)-query counting
algorithm gives a \(Q\)-query algorithm for \(\operatorname{UOR}_{N'}\).
Theorem~\ref{thm:uor-lower-bound} implies
\[
Q=\Omega(\sqrt{\zeta N'}).
\]

Under \(M+\Delta\le N-M\), we have \(M\le N/2\), so
\(N-M\ge N/2\). Also \((N-M)/\Delta>1\), and therefore
\[
N'
\ge
\frac{N-M}{2\Delta}
=
\Omega\left(\frac N\Delta\right).
\]
Substitution gives \eqref{eq:counting-reduction-bound}. Since the bound
only depends on \(N\) and \(\Delta\), it is unchanged when we translate back
from the complemented instance.
\end{proof}

\subsection{Combined small-bias lower bound}
\label{sec:final-counting}

The two preceding arguments give independent lower bounds for the same
promise problem. Combining
Theorems~\ref{thm:counting-direct-lower-bound} and
\ref{thm:counting-reduction-bound} yields the main result.

\begin{theorem}[Small-bias two-weight counting lower bound]
\label{thm:counting-final}
Let
\[
1\le M<M+\Delta<N,
\qquad
0<\zeta\le\frac12.
\]
Every quantum query algorithm that distinguishes
\[
|x|=M
\qquad\text{from}\qquad
|x|=M+\Delta
\]
with success probability at least \(1/2+\zeta\) uses
\begin{equation}
Q
=
\Omega\left(
\max\left\{
\zeta\frac{\sqrt{(N-M)(M+\Delta)}}{\Delta},
\sqrt{\frac{\zeta N}{\Delta}}
\right\}
\right).
\label{eq:final-two-regime-bound}
\end{equation}
If \(\Delta=\epsilon M\), this is equivalently
\begin{equation}
Q
=
\Omega\left(
\max\left\{
\frac{\zeta}{\epsilon}
\sqrt{\frac{(1+\epsilon)(N-M)}{M}},
\sqrt{\frac{\zeta N}{\epsilon M}}
\right\}
\right).
\label{eq:final-two-regime-epsilon}
\end{equation}
For \(0<\epsilon\le1\) in the convention
\(M+\Delta\le N-M\), the first term simplifies to
\(\Theta((\zeta/\epsilon)\sqrt{N/M})\).
\end{theorem}

\begin{acknowledgments}
This project is supported by NSTC under project number 115-2112-M-007 -006 -  and 114-2221-E-007 -083 -MY3, and funding of Taiwan Centers of Excellence (TCE)

OpenAI ChatGPT (GPT-5.6) was used during the development and
preparation of this work to assist with exploring proof strategies,
deriving intermediate mathematical arguments, checking calculations,
organizing the literature, and revising the presentation of the
manuscript. Some proof ideas and intermediate derivations were
initially suggested by the AI tool. The authors independently
examined, reconstructed, and verified all arguments and proofs,
checked the relevant claims and citations against the literature,
and take full responsibility for all results and conclusions.
\end{acknowledgments}

\section*{Data Availability}
No data were created or analyzed in this study.

\bibliography{references}

\end{document}